\documentclass[10pt,draft]{amsart}
\usepackage{amsmath,amssymb,amsthm}
\begin{document}
\newtheorem{lem}{Lemma}[section]
\newtheorem{prop}{Proposition}[section]
\newtheorem{cor}{Corollary}[section]
\numberwithin{equation}{section}
\newtheorem{thm}{Theorem}[section]
\theoremstyle{remark}
\newtheorem{example}{Example}[section]
\newtheorem*{ack}{Acknowledgment}
\theoremstyle{definition}
\newtheorem{definition}{Definition}[section]
\theoremstyle{remark}
\newtheorem*{notation}{Notation}
\theoremstyle{remark}
\newtheorem{remark}{Remark}[section]
\newenvironment{Abstract}
{\begin{center}\textbf{\footnotesize{Abstract}}%
\end{center} \begin{quote}\begin{footnotesize}}
{\end{footnotesize}\end{quote}\bigskip}
\newenvironment{nome}
{\begin{center}\textbf{{}}%
\end{center} \begin{quote}\end{quote}\bigskip}

\newcommand{\triple}[1]{{|\!|\!|#1|\!|\!|}}
\newcommand{\xx}{\langle x\rangle}
\newcommand{\ep}{\varepsilon}
\newcommand{\al}{\alpha}
\newcommand{\be}{\beta}
\newcommand{\de}{\partial}
\newcommand{\la}{\lambda}
\newcommand{\La}{\Lambda}
\newcommand{\ga}{\gamma}
\newcommand{\del}{\delta}
\newcommand{\Del}{\Delta}
\newcommand{\sig}{\sigma}
\newcommand{\ome}{\omega}
\newcommand{\Ome}{\Omega}
\newcommand{\C}{{\mathbb C}}
\newcommand{\N}{{\mathbb N}}
\newcommand{\Z}{{\mathbb Z}}
\newcommand{\R}{{\mathbb R}}
\newcommand{\Rn}{{\mathbb R}^{n}}
\newcommand{\Rnu}{{\mathbb R}^{n+1}_{+}}
\newcommand{\Cn}{{\mathbb C}^{n}}
\newcommand{\spt}{\,\mathrm{supp}\,}
\newcommand{\Lin}{\mathcal{L}}
\newcommand{\SSS}{\mathcal{S}}
\newcommand{\F}{\mathcal{F}}
\newcommand{\xxi}{\langle\xi\rangle}
\newcommand{\eei}{\langle\eta\rangle}
\newcommand{\xei}{\langle\xi-\eta\rangle}
\newcommand{\yy}{\langle y\rangle}
\newcommand{\dint}{\int\!\!\int}
\newcommand{\hatp}{\widehat\psi}
\renewcommand{\Re}{\;\mathrm{Re}\;}
\renewcommand{\Im}{\;\mathrm{Im}\;}

\title[Max-Min characterization of the mountain pass]%
{{Max-Min characterization of the mountain pass energy level\\ 
for a class of variational problems}}

\begin{abstract}
We provide a  max-min characterization of the mountain pass
 energy level
for a family of variational problems. As a consequence we deduce
the mountain pass structure of solutions to 
suitable  PDEs, whose existence follows from 
classical minimization argument.
\end{abstract}

\author{}

\author{Jacopo Bellazzini and Nicola Visciglia}

\address{Jacopo Bellazzini\\Dipartimento di Matematica Applicata 
Universit\`a di Pisa\\
Via Buonarroti 1/C, 56127 Pisa, Italy}

\email{j.bellazzini@ing.unipi.it}

\address{Nicola Visciglia\\
Dipartimento di Matematica Universit\`a di Pisa\\
Largo B. Pontecorvo 5, 56100 Pisa, Italy}

\email{viscigli@dm.unipi.it}

\maketitle
\date{}
In the literature the existence of solutions
for variational PDE is often reduced to the existence of critical points
of functionals $F$  having the following structure
$$F(u)=T(u)-U(u)$$
with $u$ belonging to a Banach space $X$
and $T,U$ that satisfy suitable conditions.\\
A classical strategy which is very useful in order to find critical points for the functional $F$ is to look at the following minimization problem
$$\min_{U(u)=1}T(u)$$
and eventually to remove the Lagrange multiplier that appears 
by using suitable invariances of the problem. 
On the other hand when the functional $F$ shows a mountain pass geometry
it is customary to look at critical points of the unconstrained functional $F(u)$ on the whole space $X$. More specifically it is well known that a candidate critical value is given by the mountain pass energy level
\begin{equation}\label{c}
c:=inf_{\eta \in \Gamma} sup_{t\in [0,1]} F(\eta (t))
\end{equation}
where
\begin{equation}\label{pathset}
\Gamma:=\{\eta \in C([0,1]; X |\eta(0)=0, F(\eta(1))<0\}.
\end{equation}
A natural question is to understand whether or not the two forementioned
approaches provide the same solution.\\
Of course in both approaches described above the main difficulty
is related with the eventual lack of compactness respectively for the
minimizing sequences and for the Palais Smale sequences.
The main contribution of this paper is to give a max-min characterization of the mountain pass value $c$ described above under a general framework.
Roughly speaking we relate the existence of a regular path of minimizers
for
$$\min_{U(u)=\lambda}T(u),\  \lambda \in \R^+$$
with the mountain pass energy level.
As a byproduct
we shall  show
the mountain pass structure of solutions obtained via minimization approach 
to a family of PDEs with subcritical and critical Sobolev exponent. 
Finally we mention the papers \cite{JT} and \cite{OM} where the question mentioned above is studied for specific PDEs. We shall recover those results
as a consequence of our general topological result stated below.\\
\\
In order to give a concrete example in which for instance  it is not  
easy  to reveal the mountain pass structure of solutions obtained
via minimization, we consider the following equation:
\begin{equation}\label{tufello}
-\Delta_p u = |u|^{p^*-2}u \text{ on } \R^n \text{ with } 1<p<n  \text{ and } p^*=\frac{np}{n-p}
\end{equation}
It is well known since \cite{T} that the best constant is achieved
in the following continuous embedding 
$$\mathcal{D}^{1,p}(\R^n)\subset L^{p^*}(\R^n)$$
and hence via rescaling argument there is a nontrivial solution of \eqref{tufello}. As a consequence of our next abstract theorem one can deduce easily
that any solution constructed as above via a minimization procedure
is a mountain pass solution. \\ 
\\
We fix some notations. For every $\lambda \in \R$ we introduce the sets
$${\mathcal U}_\lambda:=
\left \{u\in X| U(u)=\lambda\right \}$$
and 
$$i_\lambda:=\inf_{u\in {\mathcal U}_\lambda} T(u).$$
Finally we introduce $${\mathcal M}_\lambda
:=\left \{u \in {\mathcal U}_\lambda| T(u)
=i_\lambda\right \}$$ (notice that 
it could even be the empty set).
\begin{thm}\label{astratto}
Let $X$ be a Banach space and $F=T-U$
with
\begin{equation}\label{0}
T,U\in C(X,\R), T(u)\geq 0, T(0)= U(0)=0,
\end{equation}
\begin{equation}\label{1}
\lim_{\lambda \rightarrow 0^+} i_{\lambda}=0
\end{equation}
\begin{equation}\label{2}
\lambda^*>0 \text{ where } \lambda^*:=\inf \{\lambda\in(0, \infty) \text{ such that }i_{\lambda}-\lambda \leq 0 \}
\end{equation}
\begin{equation}\label{3}
\lambda^{**}>0 \text{ where } \lambda^{**}:=\inf \{\lambda\in(0, \infty) \text{ such that }i_{\lambda}-\lambda < 0 \}
\end{equation}
\begin{equation}\label{4}
{\mathcal M}_\lambda\neq \emptyset \ \forall \lambda\in\R^+
\end{equation}
\begin{eqnarray}\label{5}
 &\hbox{ there exists a continuous map } \gamma :[0,\infty)
\rightarrow X\\
 \nonumber&\hbox{ such that } \gamma(\lambda)\in {\mathcal M}_\lambda.
\end{eqnarray}
Then $c>0$ and 
$$\max_{\lambda\in (0, \lambda^{**})} \left(\min_{u \in \mathcal{U}_{\lambda}}F(u)\right)=\max_{\lambda\in (0, \lambda^{**})} F(\gamma(\lambda))=c$$
where $c$ is defined in \eqref{c}.
\end{thm}

\begin{remark}\label{rem1}
Notice that if we assume moreover that the functional $F$ is $C^1(X)$,
then the classical deformation lemma implies that
$\gamma(\bar \lambda)$ is critical point of $F$, where
$\max_{\lambda\in (0, \lambda^{**})} F(\gamma(\lambda))=F(\gamma(\bar \lambda))$.
\end{remark}
\begin{remark}
The hypothesis \eqref{5} concerns a continuous selections of minima
with respect to the parameter $\lambda$. Typically this is the hardest condition
to be concretely checked. However when some invariance of the
variational problem is available it is possible to prove \eqref{5}.
\end{remark}

Next we give two concrete applications
of our general result, based on the following two types of invariances: 
\begin{equation}\label{resc}
\text{  {\em rescaling} } u: \rightarrow u(\frac{x}{\beta})
\end{equation}
\begin{equation}\label{contr-exp}
\text{ {\em contraction/ expansion }} u: \rightarrow \beta u(x).
\end{equation}
The model equation  for {\em rescaling} is given  by
\begin{equation}\label{sfe1}
-\Delta_p u-\mu \frac{1}{|x|^p}|u|^{p-2}u= g(u) \hbox{ on } \R^n, n\geq 3\\
\end{equation}
which is the Euler-Lagrange equation corresponding to the following functional
$$W^{1,p}(\R^n) \ni u\rightarrow F(u)=\frac{1}{p} \left(\int |\nabla u|^P-\mu \frac{1}{|x|^p}|u|^pdx\right)-\int G(u)dx$$
Next we fix the following specific framework:
\begin{eqnarray}\label{setting1}
& X=W^{1,p}(\R^{n})\\
\nonumber & T(u)=\frac{1}{p} \left(\int |\nabla u|^P-\mu \frac{1}{|x|^p}|u|^pdx\right) \ \text{ with } 0 \leq \mu <(\frac{n-p}{p})^p\\
\nonumber & U(u)=\int G(u)dx
\end{eqnarray}
Following the same arguments as in \cite{BL} for $p=2$
and \cite{OM} for $p \neq 2$ it is easy to deduce that \eqref{4} holds in the specific context given in \eqref{setting1} provided that $1<p<N$ and the nonlinearity $g$ fulfills
\begin{equation}\label{11}
g(s)\in {\mathcal C}(\R, \R) 
\hbox{ is continuous and odd; }
\end{equation}
\begin{equation}\label{12}
-\infty <\liminf_{s\rightarrow 0}
\frac{g(s)}s
\leq \limsup_{s\rightarrow 0}
\frac{g(s)}s<0;\end{equation} 
\begin{equation}\label{13}
\lim_{s\rightarrow \infty}
\frac{g(s)}{s^{p^*-1}}=0 \text{ where } p^*=\frac{np}{n-p}
\end{equation}
\begin{equation}\label{14}
\hbox{ there exists }
\xi_0>0 \hbox{ s.t. } G(\xi_0)>0.
\end{equation}
Next corollary concerns the mountain pass structure of the solutions
obtained by scaling the minimizers described above. Let us emphasize that in next corollary $F$, $c$, $\mathcal{M}_1$ and  $i_1$ are defined as in theorem \ref{astratto}
in the concrete case given by \eqref{setting1}. 
\begin{cor}\label{cor1}
Let $1<p<n$, $G$ satisfies \eqref{11}-\eqref{12}-\eqref{13}-\eqref{14}, $v\in \mathcal{M}_1$ and
$$\bar \lambda =(i_1)^{\frac{n}{p}}(\frac{n-p}{p})^{\frac{n}{p}}$$
Then 
$u=v(\frac{x}{\bar{\lambda}^{\frac{1}{n}}})$ is a solution of \eqref{sfe1}
whose energy level is given by the mountain pass energy level $c$. Moreover
the path $\gamma:(0, \infty)\rightarrow X$ given by
$$\lambda\rightarrow v(\frac{x}{\lambda^{\frac{1}{n}}})$$ satisfies 
$\max_{\lambda \in(0,\infty)} F(\gamma(\lambda))=c$.
\end{cor}
Our second application concerns the {\em contraction-expansion} scale invariance. We consider the equation
\begin{equation}\label{sfe2}
\left\{
\begin{array}{ll}
&-\Delta_p u -\mu |u|^{p-2}u=u|u|^{p^*-2}u \text{ in } \Omega \subset \R^n\ \\
&u=0 \ \text{ on }\partial \Omega
\end{array}\right.
\end{equation}
with
\begin{equation}\label{hyp2}
p^*=\frac{np}{n-p}, \Omega \text{ is open and bounded, }  1<p^2<n 
\end{equation}
\begin{equation*}
0<\mu<\mu_p
\end{equation*}
where
$$\mu_p=\inf_{u\neq 0} \frac{\int |\nabla u|^pdx}{\int |u|^p dx}$$
This equation  is the Euler-Lagrange equation corresponding to the following functional
$$W^{1,p}_0(\Omega) \ni u\rightarrow F(u)=\frac{1}{p} \left(\int |\nabla u|^P-\mu |u|^pdx\right)-\frac{1}{p^*}\int u^{p^*}dx$$
Next we fix the following specific framework:
\begin{eqnarray}\label{setting2}
&X=W^{1,p}_0(\Omega) \\
\nonumber &T(u)=\frac{1}{p} \left (\int |\nabla u|^P-\mu |u|^pdx \right)\\
\nonumber &U(u)=\frac{1}{p^*}\int u^{p^*}dx
\end{eqnarray}
The validity of hypothesis \eqref{4} in this specific framework has been checked in \cite{BN} for $p=2$ and
in \cite{GV} for $p \neq 2$. In the  next corollary $F$, $c$, $\mathcal{M}_1$ and  $i_1$ are defined as in theorem \ref{astratto}
in the concrete case given by \eqref{setting2}.  
\begin{cor}\label{cor2}
Let \eqref{hyp2} holds, $v\in \mathcal{M}_1$ and  
$\bar \lambda=\left(\frac{i_1p}{p^*}\right)^{\frac{p^*}{p^*-p}}$.
Then 
$u=\bar{\lambda}^{\frac{p}{p^*}}v$ is a solution of \eqref{sfe2}
whose energy level is given by the mountain pass energy level. Moreover
the path $\gamma:(0, \infty)\rightarrow X$ given by
$$\lambda\rightarrow \lambda^{\frac{p}{p^*}}v $$ satisfies
$\max_{\lambda \in(0,\infty)} F(\gamma(\lambda))=c$.
\end{cor} 
\begin{remark}
Let us emphasize that the mountain pass structure of solutions
to \eqref{sfe2} obtained via the minimzation procedure,can be also deduced
by using a classification theorem
of the Palais Smale sequences of the associated functional (see \cite{GV}).
However we point out that our proof follows by purely 
topological arguments. 
\end{remark}

\section{Proof of theorem \ref{astratto}}
Notice that
$(0, \infty) \ni \lambda \rightarrow I_{\lambda}:= i_{\lambda}-\lambda$
is a continuous function due to \eqref{0} and \eqref{5}.
Let $$I:=\max_{\lambda\in (0, \lambda^{**})}I_{\lambda} \hbox{ and } 
\mathcal{A}=\left\{ \lambda \in (0, \lambda^{**}) \text{ such that } I_{\lambda}=I\right\}.$$
We claim the following fact:
\\
\\
\emph{given $\eta\in \Gamma$ there is $\bar t$ such that $\eta(\bar t)\in \mathcal{U}_{\bar \lambda}$ where $\bar \lambda \in \mathcal{A}$}\\
(recall that $\Gamma$ is defined in \eqref{pathset}).
\\
\\
We show how the claim implies the theorem.
We have 
$$\max_{t \in [0,1]}F(\eta(t))\geq F(\eta(\bar t))=T(\eta(\bar t))- U(\eta(\bar t)) \geq i_{\bar \lambda}-\bar \lambda$$
hence 
\begin{equation}\label{napoli}
c\geq i_{\bar \lambda}-\bar \lambda=I=\max_{\lambda \in (0, \lambda^{**})}F(\gamma (\lambda)) 
\end{equation}
where $\gamma \in \Gamma$ is defined by \eqref{5}.
By definition of $\lambda^{**}$ there is a sequence
$\lambda_n> \lambda^{**}$ such that 
$\lambda_n\rightarrow \lambda^{**}$, 
$F(\gamma(\lambda_n))<0$ and moreover $\lim_{n\rightarrow \infty} F(\gamma(\lambda_n))=0$.
As a consequence there is $\bar n \in \N$ such that $$\sup_{\lambda\in (0, \lambda_{\bar n})} F(\gamma(\lambda))=\sup_{\lambda\in (0, \lambda^{**})} F(\gamma(\lambda)).$$
In particular, after a suitable parametrization, $\gamma:[0, \lambda_{\bar n}] \rightarrow X$ belongs to $\Gamma$ and hence 
$$c\leq \max_{\lambda \in (0, \lambda^{**})}F(\gamma (\lambda)).$$
By combining this fact with \eqref{napoli} we get
$c=\max_{\lambda \in (0, \lambda{**})}F(\gamma (\lambda)).$
\\
In order to prove the claim stated above we notice that
\eqref{2} implies
$$F(u)\geq 0 \hbox{ } \forall u\in \cup_{\lambda\in (-\infty, \lambda^{**}]} 
{\mathcal U}_{\lambda}$$
and since $\eta \in \Gamma$ necessarily
\begin{equation}\label{interm}
\eta(1)\in 
{\mathcal U}_{\lambda} \text{ with } \lambda>\lambda^{**}
\end{equation} 
Next consider 
the continuous function
\begin{equation}\label{contt}
(0, 1) \ni t \rightarrow U(\eta(t))
\end{equation}
By combining the definition of $\Gamma$ with \eqref{interm} we get 
$U(\gamma(0))=0$ and $U(\gamma(1))>\lambda^{**}$, hence by a continuity argument we have the claim since $0\leq \bar \lambda \leq \lambda^{**}$.

\section{Applications}
\begin{proof}[Proof of corollary \ref{cor1}]
By the standard Hardy inequality we have 
\begin{equation}
\left(\frac{n-p}{p}\right)^p \frac{1}{|x|^p}|u|^pdx \leq \int |\nabla u|^pdx
\end{equation}
and hence 
\begin{equation}
\int |\nabla u|^p-\mu \frac{1}{|x|^p}|u|^pdx
\end{equation}
is equivalent to the standard seminorm $ \int |\nabla u|^pdx$
provided that $\mu$ is like in the assumptions. Moreover due to
the  positivity of $\mu$ we can deduce via  rearrangement argument
that the minimizing sequences for $i_1$ can be choosen radially symmetric.
Following the same argument as in \cite{BL} for $p=2$ and more generally
for $p \neq 2$ and \cite{OM},  we have that $\mathcal{M}_1\neq \emptyset$.
In order to prove the corollary \ref{cor1} we check first that the general hypotheses
of the theorem \ref{astratto} are fulfilled in the specific framework
defined in \eqref{setting1}.
Notice that a rescaling argument
shows that
$$i_\lambda=\lambda^{1-\frac{p}{n}} i_1$$
and a family of minimizers for $i_{\lambda}$
is given by $v(\frac{x}{\lambda^{\frac{1}{n}}})$ where $v\in \mathcal{M}_1$.
Hence we have
$$i_{\lambda}-\lambda=\lambda^{1-\frac{p}{n}} i_1-\lambda$$
and thus hypoteses \eqref{1}-\eqref{2}-\eqref{3}-\eqref{5} are fulfilled.
We have
$$\lambda^*=\lambda^{**}=i_1^{\frac{n}{p}}$$
Since $\bar \lambda =(i_1)^{\frac{n}{p}}(\frac{n-p}{p})^{\frac{n}{p}}$ fulfills
$\bar{\lambda}^{1-\frac{p}{n}}i_1 - \bar \lambda=max_{\lambda \in (0, \infty)} \lambda^{1-\frac{p}{n}}i_1 -  \lambda$ 
we can conclude due to remark \ref{rem1} that $v(\frac{x}{\bar \lambda^{\frac{1}{n}}})$
is a mountain pass solution.
\end{proof}
\begin{proof}[Proof of corollary \ref{cor2}]
The fact that $\mathcal{M}_1\neq \emptyset$ is proved in \cite{BN} and
\cite{GV} respectively in the case $p=2$ and $p \neq 2$. The proof follows
 exactly like in the proof of corollary \ref{cor1} once we notice that
by a scaling argument 
$$
i_{\lambda}=\lambda^{\frac{p}{p^*}}i_1
$$
and a path of  minimizer for $i_{\lambda}$ is given by 
$(0,\infty)\ni \lambda \rightarrow \lambda^{\frac{1}{p^*}}v$.
Notice finally that  $\bar \lambda=\left(\frac{i_1p}{p^*}\right)^{\frac{p^*}{p^*-p}}$ fulfills
$\bar{\lambda}^{\frac{p}{p^*}}i_1 - \bar \lambda=max_{\lambda \in (0, \infty)} \lambda^{\frac{p}{p^*}}i_1 -  \lambda$
and then we can conclude as in corollary \ref{cor1}. 
\end{proof}



\date{}



\end{document}